\documentclass[11pt]{article}
\usepackage{amssymb}
\usepackage[margin=1.4in]{geometry}

\usepackage{graphicx}
\usepackage{color}
\usepackage{amsmath}
\usepackage{eurosym}

\newtheorem{theorem}{Theorem}

\newtheorem{example}[theorem]{Example}

\newtheorem{proposition}[theorem]{Proposition}

\newenvironment{proof}[1][Proof]{\textbf{#1.} }{\ \rule{0.5em}{0.5em}}

\def\seuro{\operatorname{\mbox{\small \euro}}}

\begin{document}

\title{Cash non-additive risk measures:\\ horizon risk and generalized entropy}
\author{Giulia Di Nunno\thanks{Department of Mathematics,
University of Oslo, P.O. Box 1053 Blindern, 0316 Oslo Norway.
Email: giulian@math.uio.no
and NHH - Norwegian School of Economics, Helleveien 30, 5045 Bergen, Norway.}
\thanks{The research leading to these results has received funding from the Research Council of Norway (RCN) within the project {\it STORM - Stochastics for time-space risk models}  (nr. 274410).}
  \and Emanuela Rosazza Gianin\thanks{Department of Statistics and Quantitative Methods,
University of Milano-Bicocca, via Bicocca degli Arcimboldi 8, 20126 Milano Italy.
Email:  emanuela.rosazza1@unimib.it. }
\thanks{This author is member of GNAMPA-INdAM, Italy, and acknowledges the financial support of GNAMPA Research Project 2024 (PRR-20231026-073916-203).}}
\date{June 24, 2024}

\maketitle
\vspace{-0.7cm}
\begin{abstract}
Horizon risk (\cite{DNRG1}) is studied in the context of cash non-additive fully-dynamic risk measures induced by BSDEs. Furthermore, we introduce a risk measure based on generalized Tsallis entropy which can dynamically evaluate the riskiness of losses considering both horizon risk and interest rate uncertainty. The new {\it q-entropic risk measure on losses} can be used as a quantification of capital requirement.

\vspace{2mm}\noindent

\vspace{2mm}\noindent
\end{abstract}

\section{Introduction}

Horizon risk is associated with the use of a risk measure designed for a long term positions when evaluating short term investments and vice versa.
This issue is particularly important, for example,  in the context of pensions and health insurance, where long-term claims are to be expected and hedged.
In this context, for instance, we know that the use of outdated mortality rate, or the choice of an incorrect cohort may lead to wrong premia evaluations with a consequent impact on capital requirements.
In \cite{DNRG1}, horizon risk has been identified using fully-dynamic risk measures and introducing the notion of {\it horizon longevity} or {\it h-longevity}, in short, as an index of quantification.
Indeed, fully-dynamic risk measures naturally take care of the horizon in an explicit form, since they show dependance on both the evaluation and the maturity times.

We work in a complete filtered probability space $(\Omega, \mathcal{F},  (\mathcal{F}_t)_{t\in [0,T]}, P)$ for a finite $T>0$.
We recall that
a {\it fully-dynamic risk measure} is a family $(\rho_{tu})_{t,u}\triangleq (\rho_{tu})_{0 \leq t \leq u \leq T}$ of risk measures indexed by two time parameters
$\rho_{tu}: L^{p}(\mathcal{F}_u) \to L^{p}(\mathcal{F}_t)$, with $p \in [1,+\infty]$,
that are monotone, convex, 
and, for $p=\infty$, continuous from below.

These were studied in~\cite{bion-nadal-di-nunno}, under the condition of $\mathcal{F}_t$-translation invariance, otherwise called {\it cash additivity}:
\begin{equation}
\label{eq: CA}
\rho_{tu} (X + m) = \rho_{tu} (X ) - m, \qquad \mbox{for any } X\in L^p(\mathcal{F}_u), m \in L^p(\mathcal{F}_t).
\end{equation}
When \eqref{eq: CA} does not hold, the risk measure is called {\it cash non-additive}. In particular, it is called {\it cash subadditive} (see~\cite{EK-rav} and also~\cite{mastrogiacomo, wang-csa}) when
\begin{equation}
\label{eq: CSA}
\rho_{tu} (X + m) \geq \rho_{tu} (X ) - m, \qquad \mbox{for any } X\in L^p(\mathcal{F}_u), m \in L^p_+(\mathcal{F}_t).
\end{equation}

To understand horizon risk, we have to recall the {\it restriction property} (see~\cite{bion-nadal-di-nunno}):
\begin{equation}
\label{eq: restriction}
\rho_{tu}(X) = \rho_{tv}(X), \qquad \text{for any } X\in L^p(\mathcal{F}_u), \quad v \geq u,
\end{equation}
which roughly speaking, means that a risky position $X$ at a short horizon $u$ is evaluated as if it was happening at a longer horizon.
On the contrary, whenever the restriction is lifted, we have an open possibility to quantify horizon risk. For this, the notion of {\it h-longevity} has been proposed in~\cite{DNRG1} in terms of a correction term:
  \begin{equation}
  \label{eq: time-value}
  \gamma(t,u,v,X) \triangleq  \rho_{tv}(X) - \rho_{tu} (X) \geq 0  \quad \mbox{ for any }   X \in L^{p}(\mathcal{F}_u), \; t \leq u \leq v.
  \end{equation}

When working with long time horizons, one can additionally recognize that the value of money varies.
 To take care of this uncertainty,
 we include explicitly discount factors based on interest rates as in~\cite{filipovic, EK-rav} and the discussion in~\cite{farkas}.
As noted in~\cite{EK-rav}, the explicit use of interest rate leads to the introduction of cash subadditive risk measures.
To explain, we deal with quantities expressed in unit of money. Thus we call $\seuro_u$ the unit of money at time $u$. Hence a financial investment available at $u$ is denoted $X\seuro_u$, where $X$ represents the size of the investment.
Also, let $(D_{tu})_{t,u}$ be the family of discount factors $D_{tu}$ on the time interval $(t,u]$. The unit of measurement for $D_{tu}$ is then $1/ \seuro_u$.
It is then assumed that $0 < d_{tu} \leq D_{tu} \seuro_u \leq 1$, for some (stochastic) lower bound $d_{tu}$.

For any cash additive fully-dynamic risk measure $(\phi_{tu})_{t,u}$ we define
\begin{equation}
\label{eq: rho}
\rho_{tu}(X) \triangleq \phi_{tu} (D_{tu} X \seuro_u) , \qquad X \in L^p(\mathcal{F}_u).
\end{equation}
Note that $\rho_{tu}$ is cash subadditive. In fact, since $D_{tu}\seuro_u \leq 1$, by monotonicity and cash additivity of $\phi_{tu}$, we have
\begin{eqnarray*}
\rho_{tu}(X+m) &=&  \phi_{tu} (D_{tu} (X + m) \seuro_u  ) \\
&\geq & \phi_{tu} (D_{tu} X \seuro_u + m )  = \phi_{tu} (D_{tu} X \seuro_u)- m =  \rho_{tu}(X)-m,
\end{eqnarray*}
for any $X\in L^p(\mathcal{F}_u)$ and $m \in L^p_+(\mathcal{F}_t)$.
This observation triggers the interest in cash subadditive and more generally cash non-additive risk measures, which will be discussed in this work.
Our goals are to investigate h-longevity and time-consistency for cash non-additive risk measures generated first by a single backward stochastic differential equation (BSDE) and then by a family of BSDEs.
In particular we want to quantify the riskiness of financial losses as a crucial input for the establishment of capital requirements considering simultaneously also horizon risk.
For this, we propose a fully-dynamic risk measure based on the Tsallis relative entropy, which is a generalization of the classical relative entropy, see~\cite{tsallis1, tsallis2}. We call it a {\it q-entropic fully-dynamic risk measure}.

The relationship between $(\rho_{tu})_{t,u}$ and $(\phi_{tu})_{t,u}$
is evident. Convexity and monotonicity are preserved as well as normalization: $\rho_{tu}(0) = 0$.

When it comes to time-consistency, we have to be more careful. We will see that some implications that are true for cash additive risk measures fail when dealing with the cash non-additive case.
Hereafter we write the different definitions of time-consistency formulated both with and without discount factors for convenience:
\begin{itemize}
  \item {\it Strong time-consistency} (or {\it recursivity}): for any $t, u, v \in [0,T]$ with $t \leq u \leq v$,
  \begin{eqnarray*}
  \label{strong tc}
&& \phi_{tu} (- D_{tu} \phi_{uv}( D_{uv} X \seuro_v) \seuro_u ) = \phi_{tv}(D_{tv} X \seuro_v) \notag \\
&&  \rho_{tu} (-\rho_{uv}(X))= \rho_{tv}(X) , \qquad \mbox{for } X \in L^{p}(\mathcal{F}_v).
  \end{eqnarray*}
    \item {\it Order time-consistency}: for any $s, t, u \in [0,T]$ with $s \leq t \leq u$,
     \begin{eqnarray*}
    \label{order tc}
\hspace{-10mm} &&     \phi_{tu}(D_{tu} X \seuro_u)=\phi_{tu}(D_{tu} Y \seuro_u), \; X,Y \in L^p(\mathcal{F}_u) \Longrightarrow \phi_{su}(D_{su} X \seuro_u)=\phi_{su}(D_{su} Y \seuro_u) \notag \\
 \hspace{-10mm}&& \rho_{tu}(X)=\rho_{tu}(Y), \quad X,Y \in L^p(\mathcal{F}_u) \Longrightarrow \rho_{su}(X)=\rho_{su}(Y) .
  \end{eqnarray*}
    \item {\it Weak time-consistency}: for any $t, u, v \in [0,T]$ with $t \leq u \leq v$, and $X \in L^{p}(\mathcal{F}_v)$,
  \begin{eqnarray*}
  \label{eq: weak-tc}
&&   \phi_{tv} ( D_{tv} (\phi_{uv}(0)-\phi_{uv}(D_{uv}X \seuro_v)) \seuro_v )= \phi_{tv}(D_{tv} X \seuro_v)  \notag\\
 && \rho_{tv} (\rho_{uv}(0)-\rho_{uv}(X))= \rho_{tv}(X) .
  \end{eqnarray*}
  \end{itemize}

The last definition gives us a hint of the fact that \emph{normalization} is an assumption in risk measures that should no be under estimated.
Indeed, the risk associated to the strategy ``do nothing'' may carry intrinsically long term risks.
For example, not to undertake some medical treatment in health and pensions contexts or, also, not to intervene in environment preservation.
In \cite{DNRG1}, fully-dynamic risk measures and horizon risk are discussed \emph{without} assumption of normalization. Here we proceed similarly.

 Our work is organized as follows.
 In Section~\ref{Sec2} we study fully-dynamic cash non-additive risk measures. We detail the relationships among the different notions of time-consistency, enhancing the role of normalization and how restriction or h-longevity intervene. For risk measures generated by a BSDE, we characterize restriction and h-longevity giving an explicit representation of $\gamma$ in~\eqref{eq: time-value}. In particular we study the cash subadditive case.
  In Section~\ref{Sec3} we introduce the q-entropic fully-dynamic risk measure to quantify the risk of losses.
In this case both cash subadditivity and h-longevity are captured.
We show that this risk measure is less conservative than the classical fully-dynamic entropic one, which is cash additive.
In Section~\ref{Sec4}, we propose a variation on the construction of fully-dynamic risk measures using a family of BSDEs. With this, the impact of h-longevity is enhanced. In this context we deal with the q-entropic case.

 \section{Time-consistency and h-longevity } \label{Sec2}

Hereafter, we investigate normalization, time-consistency and h-longevity for fully-dynamic risk measures $(\rho_{tu})_{t,u}$ that, in general, fail to be cash additive. We start with those induced by a BSDE to continue in the general case.

\subsection{Cash non-additive risk measures and BSDEs}

Let $(B_t)_{t \in [0,T]}$ be a $d$-dimensional Brownian motion and $(\mathcal{F}_t)_{t \in [0,T]}$ its $P$-augmented natural filtration. We restrict our attention on $L^2$ spaces.
We focus on fully-dynamic risk measures $(\rho_{tu})_{t,u}$ induced by BSDEs of the following form
\begin{equation} \label{eq: BSDE}
Y_t= X+ \int_t^u g(s,Y_s,Z_s) \, ds - \int_t^u Z_s \, dB_s \quad (u\in [0,T]),
\end{equation}
whose solution $(Y_t,Z_t)_t \triangleq (Y_t,Z_t)_{t \in [0,T]}$ can be seen as a nonlinear operator depending on the driver $g$ and evaluated at the final condition $X \in L^2( \mathcal{F}_u)$ (see~\cite{peng97}). In Peng's terminology,
$\mathcal{E}^g \left(X \vert \mathcal{F}_t \right)$ denotes the $Y$-component of the solution $(Y_t,Z_t)$ at time $t$ of the BSDE above, called conditional $g$-expectation of $X$ at time $t$.
Here below, we consider an adapted driver $g: \Omega \times [0,T]\times \Bbb R \times \Bbb R^d \to \Bbb R$
satisfying the \textit{standard assumptions}:
\vspace{-1mm}

\begin{itemize}
\item uniformly Lipschitz, i.e. there exists a constant $C>0$ such that, $dP \times dt$-a.e.,
\begin{equation*}
\vert g(\omega,t, y_1, z_1) - g(\omega, t , y_2, z_2)\vert \leq C (\vert y_1 - y_2 \vert + \vert z_1 - z_2\vert ),
\end{equation*}
for any $y_1, y_2 \in \Bbb R, \, z_1, z_2 \in \Bbb R^d$,
where $\vert \cdot \vert$ denotes the Euclidean norm in $\Bbb R^k$;
\item $E\left[\int_0^T \vert g(s,0,0)\vert ^2 \, ds \right]<+\infty$.
\end{itemize}
Such standard assumptions guarantee that equation~\eqref{eq: BSDE} admits a unique solution $(Y_t,Z_t)_{t}$, with $(Y_t)_{t} \in \Bbb H^2_{[0,T]} (\Bbb R)$ and $(Z_t)_{t} \in \Bbb H^2_{[0,T]} (\Bbb R^d)$, where
\begin{equation*}
\Bbb H^2_{[a,b]} (\Bbb R^k) \hspace{-1mm} \triangleq \hspace{-1mm} \Big\{
\mbox{adapted } \Bbb R^k\mbox{-valued processes } (\eta_s)_{s \in [a,b]}: \hspace{-1mm} E\Big[\hspace{-1mm}\int_a^b \hspace{-2mm} \vert \eta_s \vert^2 \, ds \Big] \hspace{-1mm} < \hspace{-1mm} \infty
\Big\}.
\end{equation*}
For the well-known relationship among BSDEs, nonlinear expectations and dynamic risk measures in the Brownian setting,  we refer to~\cite{barrieu-el-karoui, EK-rav, jiang, peng97, rg} (under the standard assumptions and on the Brownian setting) and to~\cite{barrieu-el-karoui, kolylanski} (for the non-Lipschitz case).
In particular,
\begin{equation*} 
\rho_{tu} (X)= \mathcal{E}^g \left(-X \vert \mathcal{F}_t \right), \quad X \in L^{2} (\mathcal{F}_u),
\end{equation*}
is a fully-dynamic risk measure. We recall that if $g$ does not depend on $y$, the risk measure is cash additive.
So, for cash non-additivity, we consider $g$ depending on $y$.


\begin{proposition} \label{prop: normal-restric}
a) $(\rho_{tu})_{t,u}$ is normalized if and only if $g(t,0,0)=0$ $dP\times dt$-a.e..

\noindent
b) $(\rho_{tu})_{t,u}$ has the restriction property if and only if $g(t,y,0)=0$ $dP\times dt$-a.e., for any~$y \in \mathbb{R}$.
\end{proposition}
\begin{proof}
a)
Assume that $\rho_{tu}(0)=0$ for any $0\leq t \leq u \leq T$.
As shown in the proof of the Converse Comparison Theorem~\cite[Thm.~4.1]{BCHMPeng} and~\cite[Lemma~2.1]{jiang}, we have
$$
g(t,y,z)=\lim _{\varepsilon \to 0} \frac{\rho_{t,t+\varepsilon} \left(-y-z \cdot (B_{t + \varepsilon} -B_{\varepsilon}) \right)-y}{\varepsilon}
$$
with convergence in $L^p$ with $p \in [1,2)$, for any $y \in \mathbb{R}, z \in \Bbb R^d$, and a.a. $t \leq u$.
By extracting a subsequence the convergence is $P$-a.s.
for which
$
g(t,0,0)=\lim _{\varepsilon \to 0} \frac{\rho_{t, t+\varepsilon} \left(0 \right)}{\varepsilon} =0
$
$dP\times dt$-a.e., since $\rho_{t,t+\varepsilon}$ is normalized. \smallskip
The converse implication is immediate.

\noindent
b) Assume now that the restriction property is satisfied. Then, for any $t,u,v$ such that $0\leq t \leq u \leq v \leq T$ and for any $X \in L^{2}(\mathcal{F}_u)$ and $v \geq u \geq t$, the relations $Y_t^u=\rho_{tu} (X)=\rho_{tv} (X)=Y_t^v$ hold, where
\begin{eqnarray*}
Y_t^u=\rho_{tu} (X)&=&-X+ \int_t^u g(s,Y^u_s, Z^u_s) ds - \int_t^u Z^u_s dB_s.
\end{eqnarray*}
Similar representation hold for $Y_t^v=\rho_{tv} (X)$.

In particular, we have that $Y_u^u=-X=Y_u^v$ for all $u \leq v$ and all $X \in L^2(\mathcal{F}_u)$, in view of the restriction property.
Taking the difference $Y_u^v - Y_u^u$, we obtain that
\begin{equation}\label{fv=mart}
\int_u^v g(s, Y_s^v, Z_s^v)ds = \int_u^v Z_s^v dB_s, \qquad \text{for all } 0 \leq u \leq v.
\end{equation}
Denote $M^{(v)}(u) \triangleq \int_0^u Z_s^v dB_s$ the martingale in $u\in [0,T]$, and
$f^{(v)}(u) \triangleq \int_0^u g(s, Y_s^v, Z_s^v)ds$ the process of finite variation
in $u\in [0,T]$. Then, by~\cite[Prop.~1.2 in Ch.~4]{revuz-yor}, a continuous martingale and a process of finite variation can be equal only if the martingale is constant. Hence \eqref{fv=mart}  implies that, for any $t$, the martingale $M^{(v)} $, starting from $0$ is equal to $0$ $dP\times du$-a.e..
Hence $Z_u^v= 0$, for a.a. $u \in [0,v]$, $dP$-a.s..
Going back to \eqref{fv=mart} , we also have that
$f^{(v)}(u) = f^{(v)}(v)$ for all $u \in [0,v]$, $dP$-a.s..
Taking the derivative with respect to $u$, we have then that $g(u, Y^v_u, 0) = 0$ for a.a. $u \in [0,v]$, $dP$-a.s..
In particular, $0 = g(v, Y^v_v, 0) = g(v, -X, 0)$ for all $X\in L^2(\mathcal{F}_u)$, for a.a. $v$, $dP$-a.s.. Then it follows that $g(v, y, 0) = 0$ for a.a. $v, \in [0,T], y \in \mathbb{R}$.

The converse implication was proved in~\cite{peng97} and, for cash additive risk measures, in~\cite{DNRG1}. Assume that $g(t,y,0)=0$ for any $t \in [0,T]$ and $y \in \mathbb{R}$. For any $X \in L^{2}(\mathcal{F}_u)$ and $v \geq u$, consider
\begin{eqnarray*}
\rho_{tu} (X)&=&-X+ \int_t^u g(s,Y^u_s, Z^u_s) ds - \int_t^u Z^u_s dB_s
\end{eqnarray*}
and similarly for $\rho_{tv} (X)$,
where $(Y^{X,u}_r,Z^{X,u}_r)$ (resp. $(Y^{X,v}_r,Z^{X,v}_r)$) denotes the solution corresponding to $\rho_{tu} (X)$ (resp. $\rho_{tv} (X)$) at time $r \leq u$. Since $(Y^{X,v}_r,Z^{X,v}_r)$ with
$$
Y^{X,v}_r= \left\{
\begin{array}{rl}
Y^{X,u}_r;& r \leq u \\
-X;& u < r \leq v
\end{array}
\right. \qquad
Z^{X,v} (r,s)= \left\{
\begin{array}{rl}
Z^{X,u}_r;& s \leq u \\
0;& u < s \leq v
\end{array}
\right.
$$
is a solution of $\rho_{tv} (X)$ when $g(t,y,0)=0$ for any $t \in [0,T]$ and $y \in \mathbb{R}$, the restriction property follows by the uniqueness of the solution.
\end{proof}

\medskip
Note that Prop.~\ref{prop: normal-restric} holds for any cash non-additive fully-dynamic risk measures.
Then we can give a qualitative conclusion that there exist only three possibilities:
\begin{itemize}
\itemindent=-10pt
\item[i)] if $g(t,y,0)=0$ for any $t, y$, then $(\rho_{tu})_{t,u}$ is both normalized and restricted;\vspace{-2mm}
\item[ii)] if $g(t,y,0) \neq 0$ for some $t $ and $y \neq 0$, $(\rho_{tu})_{t,u}$ is normalized but not restricted;\vspace{-2mm}
\item[iii)] if $g(t,0,0)\neq 0$ for any $t $, $(\rho_{tu})_{t,u}$ is neither normalized nor restricted.
\end{itemize}
We stress that in i) $(\rho_{tu})_{t,u}$ is necessarily cash additive. In fact, by~\cite{BCHMPeng} Remark after Lemma 4.5, this condition together with Lipschitz assumption and convexity of $g$ in $(y,z)$ implies that $g$ is independent of $y$. In other words, fully-dynamic risk measures, restricted and induced by a BSDE with Lipschitz driver, are necessarily cash additive.


 \subsection{Cash non-additivity and time-consistency}

Let $(\rho_{tu})_{t,u}$ be a cash non-additive fully-dynamic risk measure.
With the same arguments of~\cite{DNRG1}, we see that strong implies order time-consistency.

\begin{proposition}\label{prop: tc}
\noindent a) Weak time-consistency implies order time-consistency.

\noindent b) Under normalization and restriction: strong is equivalent to weak time-consistency.

\noindent c) Weak time-consistency, h-longevity, and $\rho_{tu}(0) \leq 0$ for any $t,u \in [0,T]$ with $t \leq u$, together imply sub (strong) time-consistency, i.e. for any $0 \leq s \leq t \leq u$
\begin{equation} \label{eq: sub TC-nca}
\rho_{st} (-\rho_{tu}(X)) \leq \rho_{su}(X) \quad \mbox{ for any } X \in L^{p}(\mathcal{F}_u).
\end{equation}
\end{proposition}
\begin{proof}
The proof of a) can be driven as in~\cite[Prop.~2]{DNRG1}.

\noindent b) Assume that strong time-consistency holds. For any $0\leq s \leq t \leq u$ and $X \in L^p(\mathcal{F}_u)$,
\begin{equation*}
\rho_{su}(X)=\rho_{st}(-\rho_{tu}(X))=\rho_{st}(\rho_{tu}(0)-\rho_{tu}(X))=\rho_{su}(\rho_{tu}(0)-\rho_{tu}(X)),
\end{equation*}
where the first equality is due to strong time-consistency, the second to normalization, and the latter to restriction. Weak time-consistency is therefore proved.
The converse implication can be checked similarly.

\noindent c) As in~\cite[Rk.~4]{DNRG1}, for any $0\leq s \leq t \leq u$ and $X \in L^p(\mathcal{F}_u)$
\begin{equation*}
\rho_{su}(X)=\rho_{su}(\rho_{tu}(0)-\rho_{tu}(X))\geq \rho_{su}(-\rho_{tu}(X))\geq \rho_{st}(-\rho_{tu}(X)),
\end{equation*}
where the equality comes from weak-time-consistency, while the two inequalities from monotonicity and h-longevity.
\end{proof}
\smallskip

The following example shows that order (or strong) time-consistency does not necessarily imply weak time-consistency for cash non-additive fully-dynamic risk measures.
\begin{example}[order $\nRightarrow$ weak; strong $\nRightarrow$ weak] \label{ex: order not weak}
Consider
\begin{equation*}
\rho_{tu}(X) = E_P \left[ \left. - e^{-r(u-t)} X \right\vert \mathcal{F}_t\right], \quad X \in L^P(\mathcal{F}_u),
\end{equation*}
\vspace{-1mm}
with $r>0$ being a deterministic interest rate. It is easy to check that $(\rho_{tu})_{t,u}$ is a cash subadditive and normalized fully-dynamic risk measure that satisfies order and strong time-consistency.
Nevertheless, weak time-consistency does not hold. In fact, for any $0\leq s \leq t \leq u$ and $X \in L^p(\mathcal{F}_u)$, for $r>0$ and $u>t$,
\begin{eqnarray*}
\rho_{su} (\rho_{tu}(0)-\rho_{tu}(X))&=& \rho_{su} (-\rho_{tu}(X)) \\
&=& E_P \left[ \left.  e^{-r(u-s)} E_P \left[ \left. - e^{-r(u-t)} X \right\vert \mathcal{F}_t\right] \right\vert \mathcal{F}_s\right] \\
&=& e^{-r(u-t)} \rho_{su}(X) \neq  \rho_{su}(X).
\end{eqnarray*}
\end{example}


\vspace{-2mm}
\subsection{Cash non-additivity and h-longevity}

We now investigate under which conditions on the driver h-longevity holds. Note that we implicitly improve the corresponding result established in~\cite{DNRG1} in the cash additive case, in fact we prove a necessary and sufficient condition and not only a sufficient one.

Let $t \leq u \leq v$ and $X \in L^p (\mathcal{F}_u)$. Then
\begin{eqnarray*}
\rho_{tu}(X)&=& -X+ \int_t^u g(s,Y_s^u,Z_s^u) dv - \int_t^u Z_s ^u dB_s
\end{eqnarray*}
and similarly for $\rho_{tv}(X)$.
Define now the processes
\begin{equation*}
\bar{Y}_s^u=\left\{
\begin{array}{rl}
Y_s^u;& s \leq u \\
-X;& u<s \leq v
\end{array}
\right. ;
\quad \bar{Z}_s^u=\left\{
\begin{array}{rl}
Z_s^u;& s \leq u \\
0;& u<s \leq v
\end{array}
\right. .
\end{equation*}

\begin{proposition} \label{prop: longevity-BSDE-nca}
H-longevity holds if and only if $g(s,y,0) \geq 0$ for any $s \in [0,T], y \in \mathbb{R}$.
Furthermore, in this case,  for any $t,v \in [0,T]$ with $t \leq v$, the h-longevity $\gamma$ in \eqref{eq: time-value}
$$
\gamma(t,u,v,X)=E_{\widetilde{Q}_X} \Big[ e^{\int_t^v \Delta_y g(s) ds} \int_u^v g(s,-X,0) ds | \mathcal{F}_t \Big], \quad t\leq  u\leq v, X\in L^p(\mathcal{F}_u),
$$
\noindent where $\widetilde{Q}_X$ is a probability measure on $\mathcal{Q}_{tv}$ depending on $X$ equivalent to $P$, with density
\begin{equation*}
\frac{d \widetilde{Q}_X}{dP}= \exp \left\{ - \frac 12 \int_t^v \vert\Delta_z g(s)\vert^2 ds + \int_t^v \Delta_z g(s) dB_s \right\},
\end{equation*}
with $\Delta_z g(s) = (\Delta_z^i g(s))_{i=1,...,d}$ being defined as
\begin{equation*}
\Delta_z^i g(s) \triangleq \frac{g(s,\bar{Y}_s^u,Z_s^v)- g(s,\bar{Y}_s^u,\bar{Z}_s^u)}{d(Z_s^{v,i}- \bar{Z}_s^{u,i})} 1_{\{Z_s^{v,i} \neq \bar{Z}_s^{u,i}\}},
\end{equation*}
while
\begin{equation*}
\Delta_y g(s) \triangleq \frac{g(s,Y_s^v,Z_s^v)- g(s,\bar{Y}_s^u, Z_s^v)}{Y_s^v- \bar{Y}_s^u} 1_{\{Y_s^v \neq \bar{Y}_s^u\}}.
\end{equation*}
\end{proposition}
The probability ${\widetilde{Q}_X}$ can be interpreted as an {\it h-longevity premium measure} (see~\cite{DNRG1}).

\smallskip
\noindent
\begin{proof}
Assume that $g(s,y,0) \geq 0$ for any $s \in [0,T]$ and $y \in \mathbb{R}$. We follow similar arguments as in~\cite{DNRG1}, however adapted to the cash non-additive case.
With the same notation introduced above, let us also consider
\begin{equation*}
 \widetilde{Y}_s= Y_s^v - \bar{Y}_s^u; \quad \widetilde{Z}_s= Z_s^v - \bar{Z}_s^u.
\end{equation*}
Then
\begin{eqnarray}
\widetilde{Y}_t&&=\rho_{tv}(X)-\rho_{tu}(X) \notag\\
&&=  \int_t^v \hspace{-1mm}[g(s,Y_s^v, Z_s^v)- g(s,\bar{Y}_s^u, \bar{Z}_s^u)] ds
+ \hspace{-1mm} \int_u^v \hspace{-1mm}g(s,\bar{Y}_s^u, \bar{Z}_s^u) ds
- \hspace{-1mm}\int_t^v \hspace{-1mm} [Z_s ^v- \bar{Z}_s^u] dB_s
- \hspace{-1mm}\int_u^v \hspace{-1mm} \bar{Z}_s ^u dB_s \notag \\
&&=  \int_t^v [g(s,Y_s^v, Z_s^v)- g(s,\bar{Y}_s^u, \bar{Z}_s^u)] ds - \int_t^v \widetilde{Z}_s dB_s+ \int_u^v g(s,-X,0) ds  \notag \\
&&=  \int_t^v [\Delta_y g(s) \cdot \widetilde{Y}_s + \Delta_z g(s) \cdot \widetilde{Z}_s] ds - \int_t^v \widetilde{Z}_s dB_s+ \int_u^v g(s,-X,0) ds .\label{eq: bsde-longevity-1-nca}
\end{eqnarray}
By applying Girsanov Theorem,~\eqref{eq: bsde-longevity-1-nca} becomes
\begin{equation}
\rho_{tv}(X)-\rho_{tu}(X)
= \int_t^v \Delta_y g(s) \cdot \widetilde{Y}_s  ds- \int_t^v \widetilde{Z}_s dB^{\widetilde{Q}_X}_s+ \int_u^v g(s,-X,0) ds, \label{eq: long-linear}
\end{equation}
where $B^{\widetilde{Q}_X}_s \triangleq B_s - B_t - \int_{t}^s \Delta_z g(r) \, dr$, $s \in [t,v]$, is a $\widetilde{Q}_X$-Brownian motion.
It is then well-known (see, e.g.,~\cite[Ex.~7.2]{EK-rav}) that
\begin{eqnarray*}
\gamma(t,u,v,X)=\rho_{tv}(X)-\rho_{tu}(X)
&=&E_{\widetilde{Q}_X} \left[ e^{\int_t^v \Delta_y g(s) ds} \int_u^v g(s,-X, 0) ds \Big\vert \mathcal{F}_t\right].
\end{eqnarray*}
By assumption on $g(\cdot, y, 0)$, it then follows that $\gamma(t,u,v,X)\geq 0$.

Assume now that h-longevity holds. Proceeding as before (see \eqref{eq: long-linear}),
\begin{equation*}
\widetilde{Y}_t= \int_u^v g(s,-X,0) ds + \int_t^v \Delta_y g(s) \cdot \widetilde{Y}_s  ds- \int_t^v \widetilde{Z}_s dB^{\widetilde{Q}_X}_s,
\end{equation*}
hence, by longevity,
\begin{equation*}
\widetilde{Y}_t =E_{\widetilde{Q}_X} \left[ e^{\int_t^v \Delta_y g(s) ds} \cdot \int_u^v g(s,-X, 0) ds \Big\vert \mathcal{F}_t\right] \geq 0
\end{equation*}
for any $t \leq u \leq v$ and $X \in L^p (\mathcal{F}_u)$.
Set now
\begin{eqnarray*}
\eta^F_t &=& F+ \int_t^v \left[g(s,-X,0) 1_{[u,v]}+ \Delta_y g(s) \cdot \eta^F_s \right] ds- \int_t^v Z_s^{\eta} dB^{\widetilde{Q}_X}_s \\
R^F_t &=& F+ \int_t^v  \Delta_y g(s) \cdot R^F_s  ds- \int_t^v Z_s^{R} dB^{\widetilde{Q}_X}_s
\end{eqnarray*}
for $t \leq u \leq v$ and $F \in L^p (\mathcal{F}_u)$. Consequently,
\begin{eqnarray*}
\eta^F_t &=&E_{\widetilde{Q}_X} \left[ e^{\int_t^v \Delta_y g(s) ds} \left(F+ \int_u^v g(s,-X, 0) ds \right) \Big\vert \mathcal{F}_t\right]  \\
R^F_t &=&E_{\widetilde{Q}_X} \left[ e^{\int_t^v \Delta_y g(s) ds} \cdot F \Big\vert \mathcal{F}_t\right]
\end{eqnarray*}
and, by longevity, $\eta^F_t=\widetilde{Y}_t + R^F_t \geq R_t^F$, for any $F \in L^p (\mathcal{F}_u)$.
By the Converse Comparison Theorem of BSDEs (see~\cite{BCHMPeng, jiang}), we have
\begin{eqnarray*}
&&g(s,-X,0) 1_{[u,v]}+ \Delta_y g(s) \cdot \eta^F_s \geq \Delta_y g(s) \cdot \eta^F_s \\
&&g(s,-X,0) 1_{[u,v]} \geq 0
\end{eqnarray*}
for any $\eta_t, X$ and $s\leq u \leq v$. Hence, $g(s,y,0) \geq 0$ for any $s \in [0,T]$ and $y \in \mathbb{R}$.
\end{proof}

\medskip
\noindent
The previous result reduces to~\cite[Prop.~9]{DNRG1} in the cash additive case since the driver $g$ does not depend on $y$, hence $\Delta_y g(s) \equiv 0$ for any $s$.

Observe that, by~\cite[Prop.~7.3]{EK-rav}, if $g(t,y,z)$ is convex in $(y,z)$ and decreasing in $y$, then the corresponding fully-dynamic risk measure is cash subadditive. Indeed, also the converse implication holds true (see~\cite[Prop.~20]{laeven-et-al}).

\begin{example}
Consider the driver
\begin{equation*}
g(t,y,z)= r_t y^- +z, \quad t \in [0,T], y \in \mathbb{R}, z \in \mathbb{R}^d,
\end{equation*}
where $r_t$ can be interpreted as a positive interest rate depending on time $t$.
It is immediate to see that $g(\cdot,y,z)$ is decreasing in $y$, convex and Lipschitz in $(y,z)$. The corresponding fully-dynamic risk measure satisfies cash-subadditivity, normalization (since $g(t,0,0)=0$, see Prop.~\ref{prop: normal-restric}) and h-longevity (by $g(t,y,0) \geq 0$ for any $t,y$, see Prop.~\ref{prop: longevity-BSDE-nca}).

Instead, for $\bar{g}(t,y,z)= r_t y^- +z +1$, with $y \in \mathbb{R}$, $z \in \mathbb{R}^d$,
we obtain a fully-dynamic risk measure not normalized and without restriction.
\end{example}

\section{A q-entropic risk measure on losses} \label{Sec3}

We consider the generalized entropy studied in \cite{tsallis1, tsallis2} based on the generalization of the logarithmic and exponential functions ($\ln_q$ and $\exp_q$)\footnote{
From~\cite{tsallis1}, recall that $\ln_q(x) \to \ln x$ and $\exp_q (x) \to \exp (x)$, for $q \to 1$, where
\vspace{-2mm}
$$\exp_q(X)\triangleq [1+(1-q)x]^{\frac{1}{1-q}} \,\Bigg\{
\begin{array}{l}
\mbox{for } x\geq \frac{1}{q-1}, q \in (0,1)\\
\mbox{for }x< \frac{1}{q-1}, q>1
\end{array}
; \qquad
\ln_q(x) \triangleq \frac{x^{1-q}-1}{1-q}
\,\Bigg\{
\begin{array}{l}
\mbox{for } x\geq 0, q \in (0,1)\\
\mbox{for }x>0, q>1
\end{array}
$$}.
Applications to pricing and risk measures have been recently proposed in \cite{ma-tian}.
Consider the BSDE~\eqref{eq: BSDE} with driver
\begin{equation} \label{eq: driver - entr general}
g_q(t,y,z)= \frac{q}{2} \frac{\vert z\vert^2}{1+(1-q) y}, \quad \mbox{for } q>0.
\end{equation}
Note that when $q=1$, we reduce to the classical BSDE associated to the entropic case.
Observe that, for $q\in (0,1)$, the driver $g_q$ is decreasing on $y< \frac{1}{q-1}$ and on $y> \frac{1}{q-1}$.
Also, it is convex in $(y,z)$ for $y> \frac{1}{q-1}$ (concave for $y< \frac{1}{q-1}$).
For $q>1$, the driver is increasing on $y< \frac{1}{q-1}$ and on $y> \frac{1}{q-1}$; convex in $(y,z)$ for $y< \frac{1}{q-1}$ (concave otherwise).

Referring to \cite[Prop.~3]{bahlali-et-al}, the solution exists and is unique for the terminal condition $X\in L^2(\mathcal{F}_u)$.
By the same arguments of \cite[Thm.~4.2]{ma-tian}, the solution has representation
\begin{equation} \label{eq: q-entropy representation}
Y_t= \ln_q E_P \left[\left. \exp_q(X) \right| \mathcal{F}_t \right] \triangleq \mathcal{E}^{g_q}_{tu}(X),
\end{equation}
when $q\in (0,1)$ and $X> \frac{1}{q-1}+\varepsilon$, with $Y_t>\frac{1}{q-1}+\varepsilon$, and
when $q>1$ and $X< \frac{1}{q-1}+\varepsilon$, with $Y_t<\frac{1}{q-1}- \varepsilon$ (for some $\varepsilon > 0$).
Summarizing, the risk measure
$\rho_{tu}(X) \triangleq \mathcal{E}^{g_q}_{tu}(-X)$,  $X\in L^2(\mathcal{F}_u)$ ($t \leq u$),
associated to \eqref{eq: q-entropy representation} is fully-dynamic and cash non-additive, in particular subadditive for $q\in (0,1)$. We work from now on with $q\in (0,1)$.

Observe that if $q\approx 1$ the risk measure is close to the entropic one $\rho^{entr}_{tu}$ and the risk of \emph{all} positions can be quantified. The further we depart from $1$, the domain of quantifiable positions is more and more restricted up to $X> -1 +\varepsilon$. This can be considered a serious drawback, also in view of the associated risk-acceptable set.
To be able to quantify the risk of all losses, we need to overcome this drawback. For this, we suggest the following risk measure, which we call {\it q-entropic risk measure on losses}:
\begin{equation} \label{eq: g-entropic losses}
\rho_{tu}^{q}(X) \triangleq \rho_{tu} (- (X+\beta)^-) = \mathcal{E}^{g_q}_{tu} ((X+\beta)^-),
\end{equation}
where $\beta>0$ is a given target that can be interpreted as an acceptable loss level. Here $- (X+\beta)^-\leq 0$ represents then the loss exceeding $\beta$.
Observe that the risk of all positions $X$ can be quantified, in fact $(X+\beta)^- \geq 0 > \frac{1}{q-1} + \varepsilon$ (for some small $\varepsilon>0$).
 Since $g_q\geq 0$, the Comparison Theorem guarantees $\rho_{tu} (- (X+\beta)^-) \geq 0$, which gives a margin for hedging the loss.
The monotonicity and convexity of $\rho_{tu}^{q}$ are implied by those of $\rho_{tu}$ together with the increasing monotonicity and concavity of $-(x + \beta)^-$.

We investigate now the sensitivity of $\rho^q_{tu}$ with respect to $q \in (0,1)$.
\begin{proposition}
For any $X \in L^2(\mathcal{F}_u), \beta \in \mathbb{R}$, the q-entropic risk measure on losses $\rho^q_{tu}$ is increasing in $q$ with
\begin{equation*}
E_P[ -(X+\beta)^- \vert \mathcal{F}_t]  = \rho_{tu}^0(X)
\leq \rho_{tu}^q(X) \leq \rho_{tu}^1(X) =  \rho_{tu}^{entr}(-(X+\beta)^-).
\end{equation*}
\end{proposition}
In other words, when considered on losses, the classical entropic risk measure is more conservative than any q-entropy.\\
\begin{proof}
It is easy to check that
\begin{equation*}
\frac{\partial g_q}{\partial q}(t,y,z)=\frac{1}{2} \frac{z^2 (1+y)}{(1+(1-q) y)^2} \geq 0 \mbox{ for any } t \in [0,T], y \geq -1, z \in \mathbb{R}.
\end{equation*}
Increasing monotonicity of $\rho_{tu}^q(X)$ in $q$ follows from the Comparison Theorem of BSDEs and the increasing monotonicity of $g_q(t,y,z)$ in $q$.
\end{proof}

\vspace{2mm}
Notice that ~\eqref{eq: q-entropy representation} directly gives that $\rho_{tu}^q$ satisfies normalization and restriction.
To deal with horizon risk, replace the driver $g_q$ in~\eqref{eq: driver - entr general} with
\begin{equation*}
\tilde{g}_q(t,y,z)= \frac{q}{2} \frac{z^2}{1+(1-q) y} + a(t)
\end{equation*}
for some deterministic $a(t)\geq 0$, for all $t$, and consider the associated BSDE \eqref{eq: BSDE}.
Then we define the {\it translated q-entropic risk measures on losses} by
\begin{equation*}
\rho_{tu}^{q,a}(X) \triangleq \mathcal{E}_{tu}^{\tilde{g}_q}( (X+\beta)^-)= \ln_q E_P \left[\left. \exp_q \left((X+\beta)^- + \int_t^u a(s) ds \right) \right| \mathcal{F}_t \right],
\end{equation*}
which is convex, monotone, and satisfies h-longevity since $a(t) \geq 0$ for all $t$.

\section{Risk measures generated by a family of BSDEs} \label{Sec4}

Now we consider general fully-dynamic risk measures induced by a family of BSDEs of type~\eqref{eq: BSDE} with drivers $\mathcal{G}=(g_u)_{u \in [0,T]}$ depending on the time horizon $u$ of $\rho_{tu}$, Lipschitz, and convex in $(y,z)$.
To be more precise, assume that, for any $t \leq u$,
\begin{equation} \label{eq: rho-from-bsde-gt}
\rho_{tu}(X)=\rho_{tu}^{\mathcal{G}}(X) \triangleq \mathcal{E}^{g_u} (-X \vert \mathcal{F}_t), \mbox{ for any } X \in L^{2} (\mathcal{F}_u).
\end{equation}
Then $(\rho^{\mathcal{G}}_{tu})_{t,u}$ satisfies monotonicity, convexity, and continuity from above/below.
As in the cash additive case (see~\cite{DNRG1}), if $g_u(v,0,0)=0$ for any $v \leq u \leq T$, then $\rho^{\mathcal{G}}_{tu} (0)=0$ for any $t \leq u$. In general, however, this \textit{does not} guarantee the restriction property.

The following result characterizes when the restriction property holds.
\begin{proposition} \label{prop: restriction-BSDE-family-nca}
The fully-dynamic risk measure $ (\rho_{tu})_{t,u}$ satisfies the restriction property if and only if $g_u$ is constant in $u$ and $g_u(t,y,0)=0$ for any $t \leq u $ and $y \in \mathbb{R}$.
\end{proposition}
\noindent
\begin{proof}
If $g_u$ is constant in $u$, the restriction property~\eqref{eq: restriction} follows directly by Prop.~\ref{prop: normal-restric}.

Conversely, assume that the restriction property holds, i.e. $\rho_{tu}(X)=\rho_{tv}(X)$ for any $t \leq u \leq v$ and $X \in L^2(\mathcal{F}_u)$.
Proceeding as in the proof of the Converse Comparison Theorem of~\cite[Thm.~ 4.1]{BCHMPeng} and~\cite[Lemma~ 2.1]{jiang} and by the restriction property,
\begin{eqnarray*}
g_u(t,y,z)&=&\lim _{\varepsilon \to 0} \frac{\rho_{tu} \left(-y-z \cdot (B_{t + \varepsilon} -B_{\varepsilon})  \right)-y}{\varepsilon}=\lim _{\varepsilon \to 0} \frac{\rho_{t,t+\varepsilon} \left(-y-z \cdot (B_{t + \varepsilon} -B_{\varepsilon})  \right)-y}{\varepsilon}
\end{eqnarray*}
with convergence in $L^p$ with $p \in [1,2)$, for any $y \in \mathbb{R}, z \in \Bbb R^d$, $t \leq u$.
Taking a subsequence, we obtain that
\begin{equation*}
\frac{\rho_{tu} \left(-y-z \cdot (B_{t + \varepsilon} -B_{\varepsilon}) \right)-y}{\varepsilon} \longrightarrow g_u(t,y,z), \quad
\varepsilon \to 0, \quad  P-a.s.
\end{equation*}
We proceed similarly for $g_v$.
By restriction, we have
\begin{eqnarray*}
g_u(t,y,z)&=&\lim _{\varepsilon \to 0} \frac{\rho_{tu} \left(-y-z \cdot (B_{t + \varepsilon} -B_{\varepsilon}) \right)-y}{\varepsilon} \\
&= & \lim _{\varepsilon \to 0} \frac{\rho_{tv} \left(-y-z \cdot (B_{t + \varepsilon} -B_{\varepsilon}) \right) -y}{\varepsilon}
= g_v(t,y,z)
\end{eqnarray*}
for any $u\leq v$, with $P$-a.s. convergence. This proves that $g_u$ is constant in $u$.
By Prop.~\ref{prop: normal-restric}, the condition $g_u(t,y,0)=0$ should hold for any $t \leq u \leq T$ and $y \in \mathbb{R}$.
\end{proof}

Similarly to the cash additive case (see~\cite{DNRG1}), time-consistency and h-longevity are related to the monotonicity of $\mathcal{G}$. We recall that by \textit{increasing family} $\mathcal{G}=(g_u)_{u \in [0,T]}$ it is meant, for any $t \leq u$,
\begin{equation*}
g_t(v,y,z) \leq g_u(v,y,z)  \mbox{ for any } v \in [0,t], y \in \Bbb R, z \in \Bbb R^d.
\end{equation*}

\begin{theorem} \label{thm: sub-tc-bsde-gt-nca}
Let $(\rho_{tu})_{t,u}$ be the fully-dynamic risk measure in~\eqref{eq: rho-from-bsde-gt}.

\noindent
a) The family $\mathcal{G}$ is increasing if and only if $(\rho_{tu})_{t,u}$ satisfies sub time-consistency in \eqref{eq: sub TC-nca}.

\noindent
b) $\mathcal{G}= \{ g \}$  if and only if $(\rho_{tu})_{t,u}$ satisfies strong time-consistency.

\noindent c) If $\mathcal{G}$ is increasing and $g_u\geq 0$ for any $u \in [0,T]$, then $(\rho_{tu})_{t,u}$ satisfies h-longevity.
\end{theorem}
We omit the proof which is similar to the one of~\cite[Thm.~13, Prop.~16]{DNRG1}.

Inspired by Sec.~\ref{Sec3}, we can consider the family $\mathcal{G}$ of drivers
\begin{equation*}
\tilde{g}_{q,u}(t,y,z)= \frac{q_u}{2} \frac{z^2}{1+(1-q_u) y} + a_u(t) \qquad\text{with } a_u(t) \geq 0,
\end{equation*}
which leads to a {\it generalized (translated) q-entropic risk measure on losses}.



\end{document}